\documentclass[conference]{IEEEtran}

\usepackage{color}
\usepackage{graphicx}
\usepackage{amsfonts}
\usepackage{amsmath}
\usepackage{graphicx}
\usepackage{cite}
\usepackage{epsfig}
\usepackage{latexsym}
\usepackage{subfigure}
\usepackage{epstopdf}
\usepackage{verbatim}
\usepackage{units}
\usepackage{amsthm}
\usepackage[longend,ruled]{algorithm2e}
\usepackage{booktabs}

\newcommand{\mat}{\boldsymbol}

\newtheorem{lemma}{Lemma}

\begin{document}
\title{Energy Efficient Scheduling for Loss Tolerant IoT Applications with Uninformed Transmitter}
\author{\IEEEauthorblockN{M. Majid Butt\IEEEauthorrefmark{1},~Eduard A. Jorswieck\IEEEauthorrefmark{2},~Nicola
Marchetti\IEEEauthorrefmark{1}}
\IEEEauthorblockA{\IEEEauthorrefmark{1}CONNECT, Trinity College, University of Dublin, Ireland\\
Email: \{majid.butt,~nicola.marchetti\}@tcd.ie}
\IEEEauthorblockA{\IEEEauthorrefmark{2}Department of Electrical Engineering and Information Technology, TU
Dresden, Germany\\
Email: eduard.jorswieck@tu-dresden.de}
}
\maketitle
\begin{abstract}
In this work we investigate energy efficient packet scheduling problem for the loss tolerant applications. We
consider slow fading channel for a point to point connection with no channel state information at the
transmitter side (CSIT). In the absence of CSIT, the slow fading channel has an outage probability associated
with every transmit power. As a function of data loss tolerance parameters and peak power constraints, we
formulate an optimization problem to minimize the average transmit energy for the user equipment (UE). The
optimization problem is not convex and we use stochastic optimization technique to solve the problem. The
numerical results quantify the effect of different system parameters on average transmit power and show
significant power savings for the loss tolerant applications.
\end{abstract}

\section{Introduction}
Internet of things (IoT) is one of the use cases of 5G wireless communications to serve the heterogeneous
services. The applications like smart city, smart buildings and smart transportation systems depend heavily on
efficient information processing and reliable communication techniques. The use of thousands of smart and tiny
sensors to communicate regular measurements, e.g., temperature, traffic volume, etc., makes it extremely
important to look at the energy efficiency aspect of the problem. In 5G networks, context aware scheduling is
believed to play key role in smart use of resources. Depending on the application's context, it may not be
necessary to receive every packet correctly at the receiver side to avoid experiencing a serious degradation in
quality of experience (QoE). If some packets are lost, the application may tolerate the loss without requiring
retransmissions of the lost packets. The application loss tolerance can effectively be exploited to reduce
average energy consumption of the devices.

We investigate energy efficient power allocation scheme for the wireless systems with data loss constraints.
The packet loss constraints are defined in terms of average packet loss and the maximum number of packets lost
in successive time slots. The reliability aspect of the communication systems is conventionally handled at
upper layers of communication using error correction codes and/or hybrid automatic repeat request (HARQ).
Feedback based link adaptation applied in HARQ is dictated by the latency constraints of the application
\cite{Villa:2012, Choi_TVT:2013}. Our approach is different from the HARQ scheme in the sense that we assume
that we do not have a data buffer at the transmitter side due to simple nature of sensing device (node), which
makes HARQ irrelevant. Instead, we assume that the applications's QoE does not require every packet to be
received successfully, i.e., loss of some packets can be tolerated, but it must be bounded and parameterized.

In literature, some earlier works have addressed similar problems in different settings (more at network
level). In \cite{Nasralla:2014}, the authors evaluate the subjective and objective performance of video traffic
for bursty loss patterns. Reference \cite{Zou:2013} considers real-time packet forwarding over wireless
multi-hop networks with lossy and bursty links. The objective is to maximize the probability that individual
packets reach their destination before a hard delay deadline. In a similar study, the authors in
\cite{Aditya_TWC:2010} investigate a scenario where multimedia packets are considered lost if they arrive after
their associated deadlines. Lost packets degrade the perceived quality at the receiver, which is quantified in
terms of the "distortion cost" associated with each packet. The goal of the work in \cite{Aditya_TWC:2010} is
to design a scheduler which minimizes the aggregate distortion cost
over all receivers. The effect of access router buffer size on packet loss rate is studied in
\cite{Sequeira:2013} when bursty traffic is present. An analytical framework to dimension the packet loss
burstiness over generic wireless channels is considered in  \cite{Fanqqin:2013} and a new metric to
characterize the packet loss burstiness is proposed. However, these works do not characterize the effect of
average and bursty packet loss on the consumed energy at link level.

The energy aspect of the problem has been addressed in \cite{Neely2009} where the authors investigate
intentional packet dropping mechanisms for delay limited systems to minimize energy cost over fading links.
Some recent studies in \cite{majid_TWC:13, majid:sys2016} characterize the effect of packet loss burstiness on
average system energy for a multiuser wireless communication system where the transmit channel state
information (CSIT) is fully available or erroneous. This work extends the work \cite{majid_TWC:13,
majid:sys2016} such that no CSIT is assumed to be available, which poses new challenges for communication and
scheduler design. When CSIT is not available for slow fading channels, channel state dependent power control
cannot be applied and error free communication cannot be guaranteed. This results in outage which adds a new
dimension to the problem. Under different system settings, we characterize the average power consumption of the
point to point wireless network for various average and bursty packet drop parameters, as well as the outage
probability that application can tolerate loss of a full sequence of packets (successively). We model and
formulate the power minimization problem, characterize the resulting programming problem and propose a solution
based on stochastic optimization. Simulation results show that our scheduling scheme exploits packet loss
tolerance of the application to save considerable amount of energy; and thereby significantly improves the
energy efficiency of the network as compared to lossless application case.

The rest of the paper is organized as follows. The system model for the work is introduced in Section
\ref{sect:system_model} and state space description of the proposed scheme is discussed in Section \ref{sect:
state space}. We formulate the optimization problem in Section \ref{sect:optimization} and discuss the solution
in Section \ref{sect:stochastic}. We evaluate the numerical results in Section \ref{sect:results} and Section
\ref{sect:conclusions} summarizes the main results of the paper.

\section{System Model}
\label{sect:system_model}
We consider a point-to-point system such that the transmitter user equipment (UE) has a single packet to
transmit in each time slot. The packets are assumed to be with fixed size, measured in bits/s/Hz. Time is
slotted and the UE experiences quasi-static independently and identically distributed (i.i.d) block flat-fading
such that the fading channel remains constant for the duration of a block, but varies from block to block.

We assume that no transmit channel state information (CSIT) is available at the transmitter, but the
transmitter is aware of channel distribution. Depending on the scheduling state $i$ (explained later in Section
\ref{sect: state space}), the UE transmits with a fixed power $P_i$ to transmit a fixed size packet with rate
$R$ bits/s/Hz, and waits for the feedback. For convenience, the distance between the transmitter and the
receiver is assumed to be normalized.

For a transmit power $P_i$, and channel fading coefficient $h$, the outage probability for the failed
transmission (channel outage) is denoted by $\epsilon_i$ such that,
\begin{equation}\label{eqn:outage}
  \epsilon_i=\Pr\left[\log_2\left(1+\frac{P_i|h|^2}{N_0}\right)<R\right]
\end{equation}
where $N_0$ is additive white Gaussian noise power.

If the packet is received at the receiver correctly, the receiver sends back a positive acknowledgement (ACK)
message to the UE. If it is not decoded at the receiver, a negative acknowledgement (NAK) is fed-back to the
UE. The feedback is assumed to be perfect without error. Note that a power adaptation based on the feedback
results is applied even without CSIT.

Feedback based power allocation belongs to Restless Multi-armed Bandit Processes (RMBPs) \cite{Whittle:1988}
where the states of  the UE in the system stochastically evolve based on the current state and the action
taken. The UE receives a reward depending on its state and action. The next action depends on the reward
received and the resulting new state. In this work, we investigate the effect of feedback based sequential
decisions in terms of UE consumed average power.

\subsection{Problem Statement}
A single packet arrives at the transmit buffer of the UE in every time slot. The UE's data buffer has no
capacity to store more than one packet ($R$ bits/s/Hz). This is a typical scenario for a wireless sensor
network application where data measurements arrive constantly after regular fixed time intervals. The UE is
battery powered, which needs to be replaced after regular intervals. It is therefore, important to save
transmit energy as much as possible. Depending on the application, the UE has two constraints on reliability of
data packet transfer \cite{majid_TWC:13, majid:sys2016}:
\begin{enumerate}
  \item Average packet drop/loss rate $\gamma$ is the parameter that constraints the average number of
      packets dropped/lost.
  \item Maximum number of packets dropped successively. This is called bursty packet drop constraint. The
      parameter $N$ denotes the maximum number of packets allowed to be dropped successively without
      degrading QoE below a certain level. Mathematically, the distance $r(q,q-1)$ between $q^{th}$ and
      $q^{th}-1$ correctly received packets measured in terms of number of packets is constrained by
      parameter $N$, i.e.,
      \begin{equation}
        r(q,q-1)\leq N.
        \label{eqn:successive}
      \end{equation}
\end{enumerate}
Due to transmit power constraint, it is not possible to provide the guarantee in (\ref{eqn:successive}) with
probability one. Given at least $N$ packets have been lost successively by time instant $t-1$, we define a
parameter $\epsilon_{out}$ at an instant $t$ by the probability that another packet is lost, i.e.,
      \begin{equation}
        \epsilon_{out}=\Pr\Big(r_t (q,q-1) = r_{t-1}(q,q-1)+1|r_{t-1}(q,q-1)\geq N\Big)
        \label{eqn:bursty_loss}
      \end{equation}

All of these factors contribute to the QoE for the application. Average packet drop rate is commonly used to
characterize a wireless network and bounds the QoE for the application. However, the bursty packet loss in the
applications like smart monitoring sensors can degrade the performance enormously due to absence of contiguous
data measurements. At the same time, the UE can exploit the parameters $\gamma$ and $N$ to optimize average
energy consumption if the application is more loss tolerant. If the application is loss tolerant, it is
advantageous to transmit with a small power if a packet has just been received successfully in the last time
slot because the impact of packet loss due to outage is not so severe on cumulative QoE. The consideration of
bursty (successive) packet loss poses a new challenge in system modeling as the number of packets lost in
previous time slots affect the power allocation decision at time slot $t$.

Clearly, there is a trade-off between transmitting a packet at time $t$ with small power based on the success
of transmission in time slots $[t-1, t-2,\dots]$, and transmitting with large power to limit the risk of
outage. This trade-off determines the power allocation policy. Let us illustrate the impact of ACKs and NAKs on
the tightness of the constraints in the following:

If the permitted average packet loss rate $\gamma$ is very high but $N$ is small, i.e., it is not permitted to
lose more than $N$ packets successively without degrading QoE, the effective average packet drop rate becomes
much lower than the permitted $\gamma$ in this case. It may work to transmit with small power due to large
$\gamma$, but parameter $N$ does not allow it.\footnote{The effect of both parameters has been characterized in
\cite{majid_TWC:13}.} Due to successive packet drop constraint $N$, transmission of a packet in a time slot $t$
may not be as critical as in any other time slot with $t'\ne t$. If a packet was transmitted successfully in a
time slot $t-1$, it implies that transmitting a packet with a lower power is not as risky in time slot $t$.
However, when the number of successively lost packets approach $N$, power allocation needs to be increased
proportionally to avoid/minimise the event of missing $N$ packets successively, which may cause loss of
important information for wireless sensor networks.

\section{State Space Description}
\label{sect: state space}
To model the problem, we need to take the history of transmission in the last $N$ time slots into account. If a
NAK is received in time slot $t-1$, it needs to be determined whether transmission in time slot $t-2$ was an
ACK or NAK. We model the problem using a Markov chain model where the next state only depends on the current
state and is independent of the history. A Markov state $i$ is defined by the number of packets lost
successively at the transmit time $t$. If a packet was transmitted successively in time slot $t-1$, the current
state $i=0$. If two successive packets are lost in time slots $t-1$ and $t-2$, $i=2$. The maximum number of
Markov states is determined by parameter $N$.

To explain the state transition mechanism, let us examine the power allocation policy first.
At the beginning of the Markov chain process, a packet is transmitted with power $P_0$ in a time slot $t$ with
initial state $i=0$. The channel has an outage probability of $\epsilon_i$ (defined in (\ref{eqn:outage})).
If the received feedback is ACK, the process moves back to state $0$, otherwise moves to state $1$. The lost
packet is dropped permanently as UE has no buffer. In state $i=1$, the new arriving packet is transmitted with
a power $P_1>P_0$ as the packet is more important for QoE at the receiver end due to previously lost packet in
the last time slot. Thus, power allocation in state $i$ is a function of outage probability $\epsilon_i$ in
state $i$,
\begin{equation}
P_i=f(\epsilon_i)
\end{equation}
If the packet is transmitted successfully, the next state is zero, 2 otherwise. Similarly, the Markov chain
makes a transition to either state $i+1$ or state zero corresponding to the event of unsuccessful or successful
transmission, respectively.
When $i=N$ (termination state) and a packet is not transmitted successfully, this defines the outage event for
successive packet loss. This is modeled by self state transition probability $\alpha_{NN}$ of staying in Markov
state $S_N$ such that,
\begin{eqnarray}
  \alpha_{NN}=\epsilon_N
  =\Pr(S_{t+1}=N|S_t=N).
\end{eqnarray}
$P_N$ is chosen such that $\alpha_{NN}\leq \epsilon_{out}$ where $\epsilon_{out}$ is a system parameter defined
in (\ref{eqn:bursty_loss}). If a packet is lost in state $N$, we want Markov process to stay in state $N$ for
the next time slot to maximize the chances of transmission for the next packet as state $N$ has the largest
transmit power $P_N$.
\begin{lemma} For all $i \in [0,N]$ it holds
 $P_i\leq P_{i+1}$.
 \label{lem:power_level}
\end{lemma}
\begin{proof}
  It is straight forward to prove by contradiction. If $P_i>P_{i+1}$ and the UE is allowed to enter state
  $i+1$, an optimal decision is not to transmit in state $i$ at all and wait for a transmission in state $i+1$
  which requires less power. This is a birth death process where after every $N-1$ time slots, one transmission
  is made in state $N$ with power $P_N$. This clearly is suboptimal solution, and makes solving problem for
  most of the realistic $\gamma$ and $N$ values infeasible.
\end{proof}
The state transitions from state $i$ to $j$ occur with a state transition probability $\alpha_{ij}$.
The state transition probability is a function of parameters $\gamma, N$ and channel distribution. For every
transmit power $P_i$, there is an associated state transition probability $\alpha_{ij}$.

\begin{figure}
\center
\includegraphics[width=3.5in]{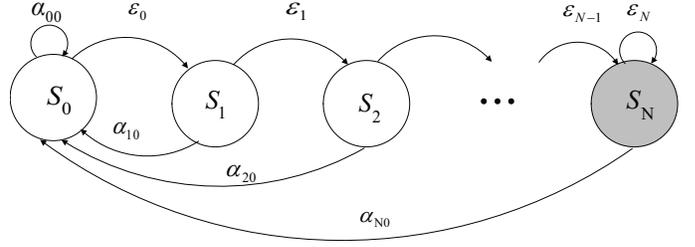}
\caption{State diagram for the Markov chain for the UE power allocation scheme.}
\label{fig:state_dia}
\end{figure}

Formally, the state transition probability $\alpha_{ij}$ from the current state $S_t=i$ to next state
$S_{t+1}=j$ is defined by,
\begin{eqnarray}
\alpha_{ij} &=& {\rm Pr}(S_{t+1}=j|S_t=i)\\
&=&\begin{cases}
              1-\epsilon_i, & \mbox{if ACK Received},\forall i, j=0  \\
              \epsilon_i, & \mbox{if NAK Received},i\ne N,j=i+1,\\&0 \leq \epsilon_i \leq 1 \\
              \epsilon_N,&\mbox {if NAK Received},i= N,j=N  \\
              0, & \mbox{otherwise}
            \end{cases}
\end{eqnarray}
where $\epsilon_i$ is given by (\ref{eqn:outage}).
The resulting state diagram is shown in Fig. \ref{fig:state_dia}.
The state transition probability matrix $\mathbf{A} = [\alpha_{ij}]_{i,j=0}^N$ takes the form
\begin{equation}\label{}
  \mathbf{A}=\left(
               \begin{array}{ccccc}
                 1-\epsilon_0 & \epsilon_0 & 0 &\dots & 0 \\
                 1-\epsilon_1  & 0& \epsilon_1 &\dots& 0 \\
                  \ddots&\ddots&\ddots&\ddots&\ddots\\
                  1-\epsilon_{N-1} & 0 & 0 &\dots& \epsilon_{N-1} \\
                 1-\epsilon_N & 0 & 0 &\dots& \epsilon_{N} \\
               \end{array}
             \right)
\end{equation}
For a time homogeneous Markov chain, the steady state probability for state $j$, $\pi_j$ is defined by
\begin{equation}
\pi_j=\sum_{i\in \mathcal{S}}\alpha_{ij}\pi_i
\end{equation}
where $\mathcal{S}$ defines the state space for the UE states.

Assuming $N_0=1$, for Rayleigh fading and state $i$, the outage probability is given by,
\begin{eqnarray}
  \epsilon_i&=& 1-\exp\Big(\frac{-(2^R-1)}{P_i}\Big)
\end{eqnarray}
After some algebraic manipulation, the required transmit power $P_i$ is calculated by,
\begin{equation}
P_i=\frac{1-2^R}{\log(1-\epsilon_i)}
\label{eqn:power}
\end{equation}
Note that other channel distributions, e.g., diversity reception or transmission with multiple receive and/or
transmit antennas with single-stream transmission and $2 \cdot d$ ($d$ is the number of active antennas) fold
diversity, result in similar outage probability expressions, as in equation (27) in \cite{Jorswieck2005d}:
\begin{eqnarray}\label{eq:outPdiv}
\epsilon_i = P \left( d \frac{2 ^R - 1}{P_i}, d \right)
\end{eqnarray}
with the incomplete Gamma function $P(a,x)$ defined in \cite[6.5.1]{Abramowitz}. It is not easy to solve
(\ref{eq:outPdiv}) with respect to $P_i$ due to the incomplete Gamma function.

From the transmit power for every state $i$, the average transmit power consumed is given by,
\begin{equation}\label{eqn:avg_power}
P_{a}=\sum_{i=0}^N P_i \pi_i.
\end{equation}

\section{Optimization Problem Formulation}
\label{sect:optimization}
The optimization problem is to compute a vector of power values $\mathbf{P}=[P_0,P_1,\dots P_N]$, which
satisfies the constraints on packet dropping parameters and minimizes average system energy. The problem is
mathematically formulated as,
\begin{eqnarray}
\label{eqn:optimization}
&&\min_{\mathbf{P}} P_a\\
s.t.
&&\begin{cases}
\mathcal{C}_1:\gamma_r\leq \gamma,& 0\leq\gamma\leq 1\\
\mathcal{C}_2:\epsilon_N\leq\epsilon_{out}& 0\leq\epsilon_{out}\leq 1
\end{cases}
\label{eqn:constrains}
\end{eqnarray}
$\mathcal{C}_1$ is the average packet loss constraint for the achieved average packet loss rate $\gamma_r$.
From the state space model,
\begin{equation}
\gamma_r = \sum_{i=0}^{N}\epsilon_{i}\pi_i~.
\label{eqn:drop_cons2}
\end{equation}
The outage probability $\epsilon_i$ and the corresponding transmit power $P_i$ for a UE in state $i$ is
computed such that the average packet dropping probability constraint $\mathcal{C}_1$ holds.
For $i=N$, $\epsilon_N\leq \epsilon_{out}$ where $\epsilon_{out}$ is defined in (\ref{eqn:bursty_loss}).
$\epsilon_i$ cannot be determined directly and needs to be optimized for the system parameters.
\begin{equation}\label{}
\epsilon_i=f(\gamma,N,\epsilon_{out},h_X(x),R)
\end{equation}
where $h_X(x)$ is the fading channel distribution.

The optimization problem is to find $\epsilon_i$, $\forall i$ that results in minimum average power. If we
choose $P_i$ too high for small states, the packets will more likely be transmitted too early at the expense of
larger power budget without exploiting loss tolerance of the application and provide good (but unnecessary)
QoE. On the other side, if $P_i$ is chosen too low in the beginning, the packets will be lost mostly and we
have to transmit with much higher power to meet the \emph{forced} condition that at least one packet has to be
transmitted to avoid the sequence of $N$ lost packets.

\subsection{Special Case $N=1$}
\label{sect:N=1}
Let us examine a special case with $N=1$. In this case, state transition probability matrix $\mat{A}$ reads,
\begin{equation}
  \mathbf{A}=\left(
               \begin{array}{cc}
                 1-\epsilon_0 & \epsilon_0\\
                 1-\epsilon_1 & \epsilon_1\\
               \end{array}
             \right)
\end{equation}
Steady state transition probabilities for states $0$ and $1$ are calculated as,
\begin{eqnarray}
\pi_0 = \frac{1-\epsilon_1}{1+\epsilon_0-\epsilon_1} \\
\pi_1 = \frac{\epsilon_0}{1+\epsilon_0-\epsilon_1}.
\end{eqnarray}
Computing $\gamma_r$ for $\epsilon_1=\epsilon_{out}$ and $\pi_0$ and $\pi_1$ calculated above,
(\ref{eqn:drop_cons2}) yields
\begin{equation}
\gamma_r=\frac{\epsilon_0}{1+\epsilon_0-\epsilon_{out}}.
\label{eqn:gamma}
\end{equation}
We can compute the value of $\epsilon_0$ in closed form that satisfies constraint $\mathcal{C}_1$ and
$\mathcal{C}_2$ with equality. Solving (\ref{eqn:gamma}) and $\mathcal{C}_1$ in (\ref{eqn:constrains}) with
equality,
\begin{equation}
\epsilon_0= (1-\epsilon_1)\frac{\gamma}{1-\gamma}.
\end{equation}
Then, we compute power levels $P_0$ and $P_1$ and resulting average power $P_{a}$ in closed form using
(\ref{eqn:power}) in (\ref{eqn:avg_power}). We numerically show in Section \ref{sect:results} that the power
levels computed in closed form for the boundary condition $\epsilon_N=\epsilon_{out}$ is not optimal for every
value of $\epsilon_{out}$.

The expressions for the power levels cannot be obtained in closed form for $N>1$. The variables
$\epsilon_0,\epsilon_1\dots \epsilon_{N}$ are unknown and it is not possible to compute a unique set of
$\epsilon_i,\forall i$ in closed form that satisfies $\mathcal{C}_1$ in (\ref{eqn:constrains}).
The optimization problem in (\ref{eqn:optimization}) is a combinatorial problem as it is hard to compute a
unique solution in terms of $\epsilon_i,\forall i$ due to sum of product term in (\ref{eqn:drop_cons2}). It is
therefore, difficult to compute $\mathbf{P}$ that minimizes $P_{a}$ using convex optimization techniques.

\begin{figure*}
\centering
  \subfigure[Average power for both closed form solution and SA method.]
 {\includegraphics[width=3.5in]{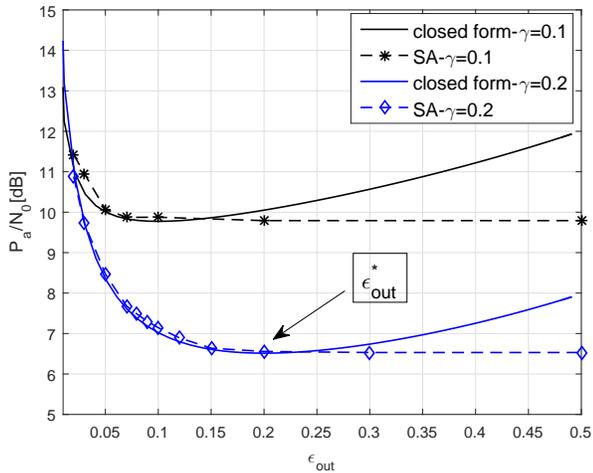}
 \label{fig:power}}
 \subfigure[Achieved $\gamma_r$ and $\epsilon_N$ for different $\epsilon_{out}$.]
{\includegraphics[width=3.5in]{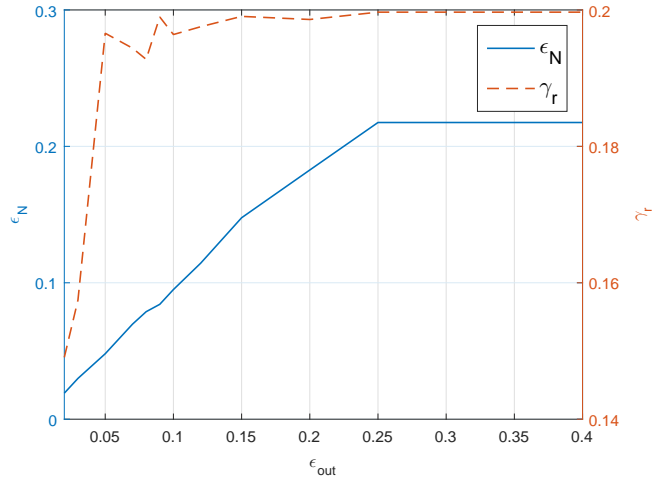}
  \label{fig:epsilon}}
 \caption{Performance of the proposed scheme for different packet loss parameters with $N=1$.}
 \vspace{-0.4cm}
 \label{fig:n1_case}
\end{figure*}

\subsection{Optimization with Peak Power Constraint}
Let us assume that we have a peak power constraint $P_m$ at the transmitter side. This implies that largest
transmit power at the UE cannot exceed $P_m$ in any state $i$, regardless of the other problem constraints.
Thus, peak power constraint is added to the constraints in (\ref{eqn:constrains}):
\begin{eqnarray}
\mathcal{C}_3: P_i\leq P_m,&\forall i,j
\label{eqn:power_const}
\end{eqnarray}
where $\mathcal{C}_3$ represents the peak power constraint.
\begin{lemma}
The peak power constraint $P_i\leq P_m, \forall i,j$, reduces to $P_N\leq P_m$.
\label{lem:peak power}
\end{lemma}
\begin{proof}
From  Lemma \ref{lem:power_level}, $P_i\leq P_{i+1}$, $\forall i$. This implies, $P_N$ is the largest transmit
power for any state. Constraining $P_N \leq P_m$ is therefore, enough to apply peak power constraint for the
overall system.
\end{proof}
From Lemma \ref{lem:peak power}, $P_N$ is constrained by $P_m$. However, $P_N$ is also constrained by the power
resulting from system parameter $\epsilon_{out}$ ($\mathcal{C}_2$). This implies that the problem is only
feasible if the solution satisfies both outage probability resulting from the peak power constraint and the
outage constraint $\epsilon_{out}$.
Denoting the power consumption from $\epsilon_{out}$ by $P_{out}$, the solution is feasible if
\begin{equation}\label{eqn:power_const}
 P_{out}\leq P_N\leq P_m.
\end{equation}

\section{Stochastic Optimization}
\label{sect:stochastic}
The combinatorial optimization problems which are not solvable with regular optimization techniques, can
approximately be solved using stochastic optimization methods. There are a few heuristic techniques in
literature to solve such problems like genetic algorithm, Q-learning, neural networks, etc.
We use Simulated Annealing (SA) algorithm to solve the problem. The algorithm was originally introduced in
statistical mechanics, and has been applied successfully to networking problems \cite{majid_TWC:13,
majid:sys2016}.

In SA algorithm, a random configuration in terms of transition probability matrix $\mat{A}$ is presented in
each iteration and the average power $P_a$ is evaluated only if constraints in (\ref{eqn:constrains}) are met.
If the evaluated $P_a$ is less than the previously computed best solution, the candidate set of outage
probabilities $\epsilon_i$, $\forall i$ are selected as the best available solution. However, the candidate set
$\epsilon_i$, $\forall i$ can be treated as the best solution with a certain temperature dependent probability
even if the new solution is worse than the best known solution. This step is called \emph{muting} and helps the
system to avoid local minima. The muting occurs frequently at the start of the process as the selected
temperature is very high and decrease as temperature is decreased gradually, where temperature denotes a
numerical value that controls the muting process.

In literature, different cooling temperature schedules have been employed according to the problem
requirements. The cooling schedule determines the convergence rate of the solution. If temperature cools down
at a fast rate, the optimal solution can be missed. On the other hand, if it cools down too slowly,
optimization requires large amount of time. In this work, we employ the following cooling schedule, called fast
annealing (FA) \cite{FA}.
In FA, it is sufficient to decrease the temperature linearly in each step $b$ such that,
\begin{equation}
\label{eqn:BA} T_b = \frac{T_0}{c_{\rm sa}\cdot b+1}
\end{equation}
where $T_0$ is a suitable starting temperature and $c_{\rm sa}$ is a constant, which depends on the
requirements of the problem. After a fixed number of temperature iterations, when muting ceases to occur
completely, the best solution is accepted as optimal solution.

\ref{algorithm}.
%
%
%

\section{Numerical Results}
\label{sect:results}
We perform numerical evaluation of the proposed scheduling scheme in this section. We consider a Rayleigh
fading channel with mean $1$ for the point to point link. The noise variance $N_0$ equals one. Spectral
efficiency $R$ equals 1 bits/s/Hz while peak power is set to a relatively high value of $20$ dBW for all
numerical examples.

%

We study the effect of packet loss parameters on average power consumption for the special case $N=1$ in Fig.
\ref{fig:n1_case}, where the results are evaluated using both closed form expressions derived in Section
\ref{sect:N=1} and the SA framework developed in Section \ref{sect:stochastic}. Average transmit power is
plotted for the fixed $N$ and $\gamma=0.1,0.2$ in Fig. \ref{fig:power}. Note that $\epsilon_{out}=\epsilon_N$
in the closed form expression. Average power consumption is a convex function in $\epsilon_{out}$ for a fixed
$\gamma$ and $N$, and a unique optimal $\epsilon_{out}$ can be seen. Let us call it $\epsilon_{out}^*$. If
system parameter $\epsilon_{out}\leq\epsilon_{out}^*$, it results in high average power. However, if
$\epsilon_{out}>\epsilon_{out}^*$, the system has more flexibility and it is optimal to set
$\epsilon_N=\epsilon_{out}^*$ instead to save power. The optimized results with SA method match closely with
the closed form results for $\epsilon_{out}\leq\epsilon_{out}^*$ which validate the accuracy of solution
provided by SA algorithm. For $\epsilon_{out}>\epsilon_{out}^*$, SA method provides the optimal solution in
contrast to the suboptimal solution where $\epsilon_N = \epsilon_{out}$ is enforced.

Fig. \ref{fig:epsilon} confirms the results in Fig. \ref{fig:power} for the experiment conducted using SA
algorithm. We plot $\epsilon_N$ (read on left y-axis) and the corresponding $\gamma_r$ values (read on right
y-axis) for different values of $\epsilon_{out}$. When $\epsilon_{out}\leq\epsilon_{out}^*$, $\epsilon_N$
follows $\epsilon_{out}$ closely while $\gamma_r<\gamma$.\footnote{The curve for $\gamma_r$ shows some
irregular behaviour. Note that $\gamma_r$ is constrained to be less than $\gamma$ and irregular values of
$\gamma_r$ resulting from stochastic optimization still meet this condition.} When
$\epsilon_{out}>\epsilon_{out}^*$, $\epsilon_N=\epsilon_{out}^*$ while $\gamma_r\to \gamma$. These results
explain the average power optimization for SA algorithm in Fig. \ref{fig:power} that all degrees of freedom are
sufficiently exploited at $\epsilon_{out}=\epsilon_{out}^*$ to optimize the energy consumption for a fixed
$\gamma$ and $N$.

Fig. \ref{fig:general_n} compares the average power consumption for the case $N=1,2,3$ and $\gamma=0.2$. The
power levels are optimized using Simulated Annealing algorithm. It is evident that $\epsilon_{out}^*$ and
resulting average power is the same for all $N$.\footnote{Minor difference in the values is due to nature of
randomized SA algorithm.} When $\epsilon_{out}\leq\epsilon_{out}^*$, an increase in $N$ for a fixed $\gamma$
helps to reduce average power consumption in general (specially at small $\epsilon_{out}$). More flexibility in
packet dropping parameters provides more degrees of freedom and results in energy savings. When
$\epsilon_{out}>\epsilon_{out}^*$, the effect of large $N$ vanishes and power saving depends solely on average
packet dropping parameter.

\begin{figure}
\center
\includegraphics[width=3.5in]{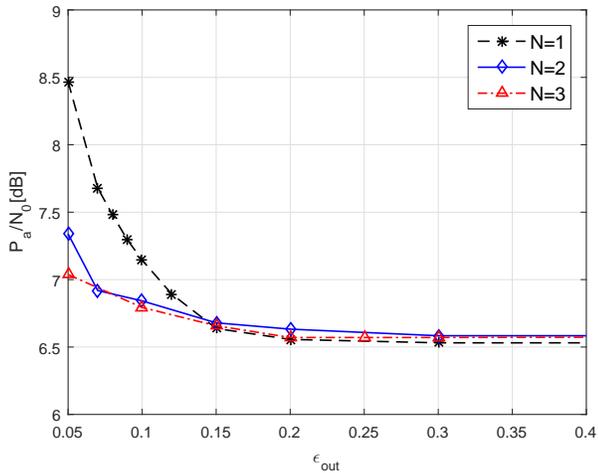}
\caption{Average power as a function of packet loss parameters for different $N$. $\gamma$ is fixed to 0.2.}
\vspace{-0.4cm}
\label{fig:general_n}
\end{figure}

\section{Conclusion}
\label{sect:conclusions}
We consider energy efficient scheduling and power allocation for the loss tolerant applications. Data loss is
characterized as a function of average and successive packet loss, and the probability that successive packet
loss is not guaranteed. These parameters jointly define the QoE and context for an application. In contrast to
average packet loss parameter, other loss parameters depend on the packet loss patterns without actually
changing the number of lost packets. By considering bursty packet loss a form of contextual information, we
provide another degree of freedom in the scheduling algorithm which can be exploited to reduce energy
consumption. Without CSIT, we formulate the average power optimization problem as a function of data loss
parameters. The optimization problem is a combinatorial optimization problem and requires stochastic
optimization technique to solve it. We compute closed form expressions of average power as a function of system
parameters for the special case $N=1$ and compare it with the solution obtained from simulated annealing
algorithm. Both of the results match up to a point and diverge after words due to inaccurate assumptions for
the closed form solution. However, the matching of both results validate the solution provided by simulated
annealing algorithm. For $N\geq1$, we numerically quantify the energy savings for increased flexibility in
successive packet loss tolerance parameter.
\section*{Acknowledgement}
This publication has emanated from research conducted with the financial support of Science Foundation Ireland
(SFI) and is co-funded under the European Regional Development Fund under Grant Number 13/RC/2077.

\bibliographystyle{IEEEtran}
\bibliography{bibliography}
\end{document}